\documentclass[prodmode,acmtalg]{acmsmall}
\usepackage{natbib}

\renewcommand{\log}{\lg}
\newcommand{\Oh}[1]
    {\ensuremath{\mathcal{O}\!\left( {#1} \right)}}
\newcommand{\rank}
    {\ensuremath{\mathrm{rank}}}
\newcommand{\select}
    {\ensuremath{\mathrm{select}}}
\newcommand{\occ}
    {\ensuremath{\mathrm{occ}}}
\newcommand{\ignore}[1]{}
    
\begin{document}

\markboth{D. Belazzougui, T. Gagie and G. Navarro}{Frequency-Sensitive Queries in Ranges}

\title{Frequency-Sensitive Queries in Ranges}
\author{DJAMAL BELAZZOUGUI \affil{Department of Computer Science, University of Helsinki}
    TRAVIS GAGIE \affil{Department of Computer Science, University of Helsinki}
    GONZALO NAVARRO \affil{Department of Computer Science, University of Chile}}

\begin{abstract}
Karpinski and Nekrich (2008) introduced the problem of parameterized range majority, which asks to preprocess a string of length $n$ such that, given the endpoints of a range, one can quickly find all the distinct elements whose relative frequencies in that range are more than a threshold $\tau$.  Subsequent authors have reduced their time and space bounds such that, when $\tau$ is given at preprocessing time, we need either $\Oh{n \log (1 / \tau)}$ space and optimal $\Oh{1 / \tau}$ query time or linear space and $\Oh{(1 / \tau) \log \log \sigma}$ query time, where $\sigma$ is the alphabet size.  In this paper we give the first linear-space solution with optimal $\Oh{1 / \tau}$ query time.  For the case when $\tau$ is given at query time, we significantly improve previous bounds, achieving either $\Oh{n \log \log \sigma}$ space and optimal $\Oh{1 / \tau}$ query time or compressed space and $\Oh{(1 / \tau) \log \frac{\log (1 / \tau)}{\log w}}$ query time.  Along the way, we consider the complementary 
problem of parameterized range minority that was recently introduced by Chan et al.\ (2012), who achieved linear space and $\Oh{1 / \tau}$ query time even for variable $\tau$.  We improve their solution to use either nearly optimally compressed space with no slowdown, or optimally compressed space with nearly no slowdown.  Some of our intermediate results, such as density-sensitive query time for one-dimensional range counting, may be of independent interest.
\end{abstract}

\category{X.0.0}{Data Structures}{How is the 2012 classification supposed to work?}

\terms{Arrays, Range Queries}

\keywords{Parameterized range majority and minority, ...}

\acmformat{Djamal Belazzougui, Travis Gagie, and Gonzalo Navarro, 2014. Frequency-Sensitive Queries in Ranges.}

\begin{bottomstuff}
This work is supported by somebody.

Authors' addresses: D. Belazzougui and T. Gagie, Department of Computer Science, University of Helsinki; G. Navarro, Department of Computer Science, University of Chile.
\end{bottomstuff}

\maketitle

\section{Introduction} \label{sec:intro}

Finding frequent elements in a dataset is a fundamental operation in data mining.  Finding the most frequent elements can be challenging when all the distinct elements have nearly equal frequencies and we do not have the resources to compute all their frequencies exactly.  In some cases, however, we are interested in the most frequent elements only if they really are frequent.  For example, Misra and Gries~\cite{MG82} showed how, given a string and a threshold $\tau$ with \(0 < \tau \leq 1\), with two passes and $\Oh{1 / \tau}$ words of space we can find all the distinct elements in a string whose relative frequencies are at least $\tau$.  These elements are called the $\tau$-majorities of the string.  Misra and Gries' algorithm was rediscovered by Demaine, L\'opez-Ortiz and Munro~\cite{DLM02}, who noted it can be made to run in $\Oh{1}$ time per element on a word RAM with \(\Omega (\log n)\)-bit words, where $n$ is the length of the string, which is the model we use; it was then rediscovered again by Karp, 
Shenker and Papadimitriou~\cite{KSP03}.  As Cormode and Muthukrishnan~\cite{CM03} put it, ``papers on frequent items are a frequent item!''

Krizanc, Morin and Smid~\cite{KMS05} introduced the problem of preprocessing the string such that later, given the endpoints of a range, we can quickly return the mode of that range (i.e., the most frequent element).  They gave two solutions, one of which takes $\Oh{n^{2 - 2 \epsilon}}$ space for any fixed positive \(\epsilon \leq 1 / 2\), and answers queries in $\Oh{n^\epsilon \log \log n}$ time; the other takes $\Oh{n^2 \log \log n / \log n}$ space and answers queries in $\Oh{1}$ time.  Petersen~\cite{Pet08} reduced Krizanc et al.'s first time bound to $\Oh{n^\epsilon}$ for any fixed non-negative \(\epsilon < 1 / 2\), and Petersen and Grabowski~\cite{PG09} reduced the second space bound to $\Oh{n^2 \log \log n / \log n}$.  Chan et al.~\cite{CDLMW12} recently gave a linear-space solution that answers queries in $\Oh{\sqrt{n / \log n}}$ time.  They also gave evidence suggesting we cannot easily achieve query time substantially smaller than $\sqrt{n}$ using linear space; however, the best known lower bound, 
by Greve et al.~\cite{GJLT10}, says only that we cannot achieve query time \(o \left( \rule{0ex}{2ex} \log (n) / \log (s w / n) \right)\) using $s$ words of $w$ bits each.  Because of the difficulty of supporting range mode queries, Bose et al.~\cite{BKMT05} and Greve et al.~\cite{GJLT10} considered the problem of approximate range mode, for which we are asked to return an element whose frequency is at least a constant fraction of the mode's frequency.

Karpinski and Nekrich~\cite{KN08} took a different direction, analogous to Misra and Gries' approach, when they introduced the problem of preprocessing the string such that later, given the endpoints of a range, we can quickly return the $\tau$-majorities of that range.  We refer to this problem as parameterized range majority.  Assuming $\tau$ is given when we are preprocessing the string, they showed how we can store the string in $\Oh{n (1 / \tau)}$ space and answer queries in $\Oh{(1 / \tau) (\log \log n)^2}$ time.  They also gave bounds for dynamic and higher-dimensional versions.  Durocher et al.~\cite{DHMNS13} independently posed the same problem and showed how we can store the string in $\Oh{n \log (1 / \tau + 1)}$ space and answer queries in $\Oh{1 / \tau}$ time.  Notice that, because there can be up to \(1 / \tau\) distinct elements to return, this time bound is worst-case optimal.  Gagie et al.~\cite{GHMN11} showed how to store the string in compressed space --- i.e., $\Oh{n (H + 1)}$ bits, where 
$H$ is the entropy of the distribution of elements in the string --- such that we can answer queries in $\Oh{(1 / \tau) \log \log n}$ time.  They also showed how to drop the assumption that $\tau$ is fixed and simultaneously achieve optimal query time, at the cost of increasing the space bound by a \((\log n)\)-factor.  That is, they gave a data structure that stores the string in $\Oh{n (H + 1)}$ space such that later, given the endpoints of a range and $\tau$, we can return the $\tau$-majorities of that range in $\Oh{1 / \tau}$ time.  Chan et al.~\cite{CDSW12} recently gave another solution for variable $\tau$, which also has $\Oh{1 / \tau}$ query time but uses $\Oh{n \log n}$ space.  As far as we know, these are all the relevant bounds for Karpinski and Nekrich's original exact, static, one-dimensional problem, both for fixed and variable $\tau$; they are summarized in Table~\ref{tab:results} together with our own results.  Related work includes Elmasry et al.'s~\cite{EHMN11} solution for the dynamic 
version and Lai, Poon and Shi's~\cite{LPS08} and Wei and Yi's~\cite{WY11} approximate solutions for the dynamic version.

\begin{table}
\tbl{Results for the problem of parameterized range majority on a string of length $n$ over an alphabet of size $\sigma$ in which the distribution of the elements has entropy $H$.\label{tab:results}}
{\begin{tabular}{l@{\hspace{2ex}}|@{\hspace{2ex}}c@{\hspace{3ex}}c@{\hspace{3ex}}c}
source & space & time & variable $\tau$\\[.5ex]
\hline\\[-2ex]
\cite{KN08} & $\Oh{n (1 / \tau)}$ words & $\Oh{(1/\tau)(\log \log n)^2}$ & no\\[1ex]
\cite{DHMNS13} & $\Oh{{n \log (1 / \tau)}}$ words & $\Oh{1 / \tau}$ & no\\[1ex]
\cite{GHMN11} & $\Oh{n (H + 1)}$ bits & $\Oh{(1 / \tau) \log \log \sigma}$ & no\\[1ex]
Theorem~\ref{thm:fixed maj} & $\Oh{n}$ words & $\Oh{1 / \tau}$ & no\\[1ex]
\hline\\[-1.5ex]
\cite{GHMN11} & $\Oh{n (H + 1)}$ words & $\Oh{1 / \tau}$ & yes\\[1ex]
\cite{CDSW12} & $\Oh{n \log n}$ words & $\Oh{1 / \tau}$ & yes\\[1ex]
Theorem~\ref{thm:fast maj} & $\Oh{n \log \log \sigma}$ words & $\Oh{1 / \tau}$ & yes\\[1ex]
Theorem~\ref{thm:small maj better} & \(n H + o (n)(H+1)\) bits & $\Oh{(1 / \tau) \log \log \sigma}$ & yes\\[1ex]
Theorem~\ref{thm:sensitive maj} & \((1 + \epsilon) n H +o(n)\) bits & $\Oh{(1 / \tau) \log \frac{\log (1 / \tau)}{\log w}}$ & yes
\end{tabular}}
\end{table}

In this paper we first consider the complementary problem of parameterized range minority, which was recently introduced by Chan et al.~\cite{CDSW12} (and then generalized to trees by Durocher et al~\cite{DSST13}).  For this problem we are asked to preprocess the string such that later, given the endpoints of a range, we can return (if one exists) a distinct element that occurs in that range but is not one of its $\tau$-majorities.  Such an element is called a $\tau$-minority for the range.  At first, finding a $\tau$-minority might seem harder than finding a $\tau$-majority because, e.g., we are less likely to find a $\tau$-minority by sampling.  Nevertheless, Chan et al. gave a linear-space solution with $\Oh{1 / \tau}$ query time even when $\tau$ is given at query time.  In Section~\ref{sec:minority} we give two results, also for the case of variable $\tau$:
\begin{enumerate}
\item for any positive constant $\epsilon$, a solution with $\Oh{1 / \tau}$ query time that takes \((1 + \epsilon) n H + \Oh{n}\) bits;
\item for any function \(f (n) = \omega (1)\), a solution with $\Oh{(1 / \tau)\,f (n)}$ query time that takes \(n H + o(n)(H+1)\) bits.
\end{enumerate}
That is, we improve Chan et al.'s solution to use either nearly optimally compressed space with no slowdown, or optimally compressed space with nearly no slowdown.  We reuse ideas from this section in our solutions for parameterized range majority.

In Section~\ref{sec:fixed maj} we return to Karpinski and Nekrich's original problem of parameterized range majority with fixed $\tau$ and give the first linear-space solution with worst-case optimal $\Oh{1 / \tau}$ query time.  In Section~\ref{sec:variable maj} we adapt this solution to the more challenging case of variable $\tau$ and give three results:
\begin{enumerate}
\item a solution with $\Oh{1 / \tau}$ query time that takes $\Oh{n \log \log \sigma}$ space, where $\sigma$ is the size of the alphabet;
\item a solution with $\Oh{(1 / \tau) \log \log \sigma}$ query time that takes \(n H + o (n)(H+1)\) bits;
\item for any positive constant $\epsilon$, a solution with $\Oh{(1 / \tau) \log \frac{\log (1 / \tau)}{\log w}}$ query time that takes \((1 + \epsilon) n H+o(n)\) bits.
\end{enumerate}
With (2), we can support $\Oh{1}$-time access to the string and $\Oh{\log \log
\sigma}$-time rank and select (see definitions in
Section~\ref{subsec:queries}); with (3), select also takes $\Oh{1}$ time.  While proving (3) we introduce a compressed data structure with density-sensitive query time for one-dimensional range counting, which may be of independent interest.
We will also show in the full version how to use our data structures for (2) or (3) to find a range mode quickly when it is actually reasonably frequent.  We leave as an open problem reducing the space bound in (1) or the time bound in (2) or (3), to obtain linear or compressed space with optimal query time.

\section{Preliminaries} \label{sec:preliminaries}

\subsection{Access, select and (partial) rank} \label{subsec:queries}

Let \(S [1..n]\) be a string over an alphabet of size $\sigma$ and let $H$ be the entropy of the distribution of elements in $S$.  An access query on $S$ takes a position $k$ and returns \(S [k]\); a rank query takes a distinct element $a$ and a position $k$ and returns the number of occurrences of $a$ in \(S [1..k]\); a select query takes a distinct element $a$ and a rank $r$ and returns the position of the $r$th occurrence of $a$ in $S$.  A partial rank query is a rank query with the restriction that the given distinct element must occur in the given position; i.e., \(S [k] = a\).  These are among the most well-studied operations on strings, so we state here only the results most relevant to this paper.

For \(\sigma = 2\) and any constant $c$, P\v{a}tra\c{s}cu~\cite{Pat08} showed how we can store $S$ in \(n H + \Oh{n / \log^c n}\) bits.  For \(\sigma = \log^{\mathcal{O} (1)} n\), Ferragina et al.~\cite{FMMN07} showed how we can store $S$ in \(n H + o (n)\) bits and support access, rank and select in $\Oh{1}$ time.  For \(\sigma < n\), Barbay et al.~\cite{BCGNN13} showed how, for any positive constant $\epsilon$, we can store $S$ in \((1 + \epsilon) n H + o (n)\) bits and support access and select in $\Oh{1}$ time and rank in $\Oh{\log \log \sigma}$ time.  Alternatively, they can store $S$ in $nH+o(n)(H+1)$ bits and support either access or select in time $\Oh{1}$, and the other operation, as well as rank, in time $\Oh{\log\log\sigma}$. Belazzougui and Navarro~\cite{BN11} showed how to support $\Oh{1}$-time partial rank using $\Oh{n (\log H + 1)}$ bits; in the full version of their paper~\cite{BN??} they reduced that space bound to \(o (n) (H + 1)\) bits.  In another paper, Belazzougui and Navarro~\cite{BN12}
 showed how, for any function \(f (n) = \omega (1)\), we can store $S$ in \(n H + o (n)(H + 1)\) bits and support access in $\Oh{1}$ time, select in $\Oh{f (n)}$ time and
rank in $\Oh{\log \log \sigma}$ time.  They also proved, via a reduction from the predecessor problem, that we cannot support general rank queries in \(o (\log (\log \sigma / \log \log n))\) time while using \(n \log^{\mathcal{O} (1)} n\) space.

\subsubsection{Alphabet partitioning} \label{subsec:alpha_partition}
The sequence representation of Barbay et al~\cite{BCGNN13} uses a technique called alphabet partitioning. 
The alphabet $[1,\sigma]$ is partitioned into at most $\lceil\log^2n\rceil$ sub-alphabets. A character $c$ occuring 
$n_c$ times will belong to sub-alphabet $\lceil\log(n/n_c)\log n\rceil$. Then the sub-alphabet mapping table $m[1..\sigma]$ 
is such that $m[c]$ stores the sub-alphabet to which character $c\in[1,\sigma]$ belongs. The sequence $t[1..n]$ over $[1..\log^2n]$ is built from $S$ by replacing every $S[i]$ by the value $m[S[i]]$ (replace every character of $S$ by the sub-alphabet it belongs to). Finally subsequences $S_i[1..t_i]$ of characters over alphabet $[1..\sigma_i]$ are built as follows. 
Let $\sigma_i$, be the number of occurrences of character $i$ in $m$. 
For every value $i\in[1..\log^2n]$ start with an empty sequence $S_i$, scan the sequence $S$ in left-to-right order and for every character $c=S[k]$ such that $m[c]=i$ append $c$ at the end of $S_i$. 
Answering rank, access and select queries on the original sequence $S$ is now achieved through combinations of rank, access 
and select queries on sequences $m$, $t$ and $S_i$. The sequences $m$ and $t$ are represented using zero-order compressed multi-ary wavelet tree supporting constant time rank access and selet queries~\cite{FMMN07}. The subsequences $S_i$ are represented 
either using a variant of Golynski et al.'s structure~\cite{GMR06} or Grossi et al's result~\cite{GOR10} achieving $O(\log\log\sigma)$ time for rank queries and either constant time select and $O(\log\log\sigma)$ time access (for the former) or 
$O(\log\log\sigma)$ select and constant time access (for the latter).

\subsection{Coloured range listing} \label{subsec:listing}

Motivated by the problem of document listing, Muthukrishnan~\cite{Mut02} showed how we can store \(S [1..n]\) such that, given the endpoints of a range, we can quickly list the distinct elements in that range and the positions of their leftmost occurrences therein.  This is the special case of one-dimensional coloured range listing in which the points' coordinates are the integers from 1 to $n$.  Let \(C [1..n]\) be the array in which \(C [k]\) is the position of the last occurrence of the distinct element \(S [k]\) in \(S [1..k - 1]\) --- i.e., the last occurrence before \(S [k]\) itself --- or 0 if there is no such occurrence.  Notice \(S [k]\) is the first occurrence of that distinct element in a range \(S [i..j]\) if and only if \(i \leq k \leq j\) and \(C [k] < i\).  We store $C$, implicitly or explicitly, and a data structure supporting $\Oh{1}$-time range-minimum queries on $C$ that returns the position of the leftmost occurrence of the minimum in the range.

To list the distinct elements in a range \(S [i..j]\) given $i$ and $j$, we find the position $m$ of the leftmost occurrence of the minimum in the range \(C [i..j]\); check whether \(C [m] < i\); and, if so, output \(S [m]\) and $m$ and recurse on \(C [i..m - 1]\) and \(C [m + 1..j]\).  This procedure is online --- i.e., we can stop it early if we want only a certain number of distinct elements --- and the time it takes per distinct element is $\Oh{1}$ plus the time to access $C$.

Suppose we already have data structures supporting access, select and partial rank queries on $S$, all in $\Oh{t}$ time.  Notice \(C [k] = S.\select_{S [k]} \left( \rule{0ex}{2ex} S.\rank_{S [k]} (k) - 1 \right)\), so we can also support access to $C$ in $\Oh{t}$ time.  Sadakane~\cite{Sad07} and Fischer~\cite{Fis10} gave $\Oh{n}$-bit data structures supporting $\Oh{1}$-time range-minimum queries.  Therefore, we can implement Muthukrishnan's solution using $\Oh{n}$ extra bits such that it takes $\Oh{t}$ time per distinct element listed.
\subsection{Minimal perfect hashing} \label{subsec:mphf}
Given a set $X\subset [1..U]$ such that $|X|=n$, a minimal perfect hash function (mphf for short)
is a bijective function from $X$ onto $[1..n]$. 
It is well-knwon result that any general scheme capable of representing 
an mphf for any given subset $X$ of $U$ of size $n$ requires exactly $n\log_2 e+\Theta(\log\log U)$ bits~\cite{Fr84}
to represent such an mphf. This bound is almost achieved in \cite{HT01} with a randomized
linear time construction and space $n\log_2 e+\Theta(\log\log U)+o(n)$ bits.  

\subsection{Monotone minimal perfect hashing} \label{subsec:mmphf}
A monotone minimal perfect hash function (mmphf) is a mphf 
which in addition to being bijective, is also monotone. 
That is given a set $X\subset [1..U]$ such that $|X|=n$, an mmphf 
$f$ over the set $X$ is a bijective function from $X$ onto $[1..n]$
and such for any pair $(x,y)\in X^2$ we have that $f(x)<f(y)$ 
if and only if $x<y$. In~\cite{BBPV09} two general schemes for 
generating mmphf representations were proposes. 
The first one allows query time $O(1)$ and representation 
space $O(n\log\log U)$ bits. The second allows query time $O(\log\log U)$
and uses space $O(n\log\log\log U)$.

\subsection{Prefix sum data structures} \label{subsec:prefsum_ds}
Given an array $A$ of $n$ values that sum up to $U$, a prefix-sum
data structure uses answers to the following queries:
given index $i$, return the sum 
of all the values of indices ranging from $1$ to $i$. 
It is possible to get a prefix sum that uses $(n+U)(1+o(1))$ bits 
of space and that answers to queries as follows. Create a bivector $V$
that contains $n$ ones and $U$ zeros by scanning the original array 
and for each value $A[i]$ append $A[i]$ zeros followed by a one. 
Then the prefix sum up to position $i$ can be answered by a $\select_1(i)$ query on the vector $V$, 
fininding the position of the $j$th one in the bitvector. 
The answer is then the number of zeros that precede that position 
which is then $j-i$. 
It is possible to improve the space of the above scheme to use only 
$n(2+\lceil \log(U/n)\rceil)+o(n)$ bits of space as follows. Build a bitvector 
that contains $n+U/\lceil U/n\rceil\leq 2n$  bits, where the $i$th one 
in the bitvector is preceded by $\lfloor((\sum_{1\leq j\leq i}A[j])/\lceil U/n\rceil )\rfloor$ zeros
We then build another vector $B$ of $n$ values of $\lceil \log(U/n)\rceil$ 
bits each, where $B[i]$ stores the value 
$(\sum_{1\leq j\leq i}A[j])-\lfloor((\sum_{1\leq j\leq i}A[j])/\lceil U/n\rceil )\rfloor$. 
In other words, the bitvector $V$ stores the prefix sums divided by 
$\lceil U/n\rceil$ and the vector $B$ stores the values modulo $\lceil U/n\rceil$. 
A query is ansered by using a select operaton on $V$ followed by 
reading one cell from $B$. 
\subsection{Indexable dictionaries} \label{subsec:indexable_dict}
Given a set $X\subset [1..U]$ with $|X|=n$, an indexable dictionary ~\cite{RRR07} is data
structure that uses $O(n\log(U/n))$ bits of space and that answers to membership queries 
in constant time. A membership query asks given any element $x\in U$, whether $x\in X$ 
or not. In addition the dictionary associates a unique number to each element in $[1..n]$
to each element of $X$. 
\subsection{Predecessor data structures} \label{subsec:pred_struct}
Given a set $X\subset [1..U]$ with $|X|=n$, a predecessor 
data structure answers to the following query. Given an element 
$y\in U$ return the greatest element $x\in X$ such that 
$x\leq y$. The y-fast trie~\cite{Wi83} achieves linear space $O(n\log U)$
bits with query time $O(\log\log U)$. The rank of $y$ is defined as the 
number of elements of $X$ no greated than $y$. 
The rank of $y$ being the same as that of its predecessor $x$, 
it is easy to modify a (static) predecessor data structure to return 
the rank of the queried element by explicitly storing the rank of every 
element $x\in X$ and return it whenever element $x$ is returned.

\ignore{OJO MORE things to the basic discussion: mmphf (incl space), prefix sum using
bits (by appending the cardinalities of the intervals in a bitmap $B_\ell$ in
unary (i.e., cardinality $c$ is stored as $1^{c-1}0$) and using $\select$ to
get prefix sums, incl space), predecessor ds incl y-fast trie and how it
returns ranks, mphf (no mmphf), dictionaries using $\Oh{n\log(U/n)}$ bits and
answer membership in constant time.

This paragraph is useful for related work on Sec \ref{sec:pred}:}
A {\em short-distance sensitive} predecessor data structure for the set of points that, given an element $x\in U$, returns the predecessor of $x$ in $\Oh{\log\log d}$ time, where $d$ is the minimum of the distances from $x$ to its predecessor and to its successor in $S$. 
The first such data structure was proposed by Johnson~\cite{Jo81}. 
Then Bose et al.\ \cite{BDDHM10,BKDHM12}, improved the space to $\Oh{n\log U\log\log\log U}$ bits. 
Both Johnson and Bose et al. solutions support insetions and deletions. 
Recently by Belazzougui et al.\ \cite{BBV11} proposed a more space-efficient static variant that uses space $\Oh{n\log U}$ bits 
only. 


\ignore{This is about the result in Sec \ref{sec:count2}:}
This bound can be considered as the one dimensional 
counterpart of the counting bound in~\cite{CW13}, 
where adaptive counting time $\Oh{\frac{\log(occ+1)}{\log\log(occ+1)}}$ was achieved
for the two-dimensional problem. We notice that both bounds
converge to the lower bounds for 2D and 1D range counting, which are
respectively $\Omega\left(\frac{\log(n+1)}{\log\log(n+1)}\right)$ and 
$\Omega\left(\sqrt{\frac{\log(n+1)}{\log\log(n+1)}}\right)$.
\section{Predecessors in a Range}
\label{sec:pred}

Assume we have a vector $X$ that contains $n$ elements from universe $[1,U]$ in 
sorted order. We consider the problem of, given an interval $[i,j]$ of the
universe that does contain elements of $X$, finding the predecessor of $j$ (the
answer can be arbitrary, even wrong, if the interval contains no point of $X$).
Our aim is to perform better when the range is smaller. We start with a basic
solution that requires time $\Oh{\lg\lg(j-i+1)}$ and $\Oh{n(\lg\lg U)^2}$ bits
of space. Then we use it as a building block to design a more elaborate 
variant that improves on both time and space.

\subsection{A Simple Data Structure} \label{sec:dist-sens}

Our data structure has $\log\log U$ {\em levels}. At a level $\ell$, we divide the universe into $\lceil U/2^{2^{\ell}}\rceil$ overlapping intervals, so that interval $k$ will be $[k\cdot 2^{2^{\ell}}+1,(k+2)2^{2^{\ell}}]$. We consider separately the intervals with even and odd $k$ (we call them even and odd intervals, respectively). For each of the two categories, the set of intervals will be disjoint. For each category, we use a mmphf $F_\ell$ that stores the values $k$ corresponding to nonempty intervals, and a prefix sum data structure $B_\ell$ to store the number of elements in each nonempty interval $k$. With $F_\ell$ and $B_\ell$ we map in constant time from a nonempty interval $k$ to its corresponding area in $X$ (say $p=F_\ell(k)$, then the area is $X[sum(B_\ell,p-1)+1~..~sum(B_\ell,p)]$). Since there are at most $n$ nonempty intervals of each category, $F_\ell$ uses $\Oh{n\log\log U}$ bits. Structure $B_\ell$ uses $\Oh{n}$ bits.

In addition, for each nonempty interval with more than $2^\ell$ elements, we store a {\em local predecessor search data structure (lpsds)}.
The lpsds of an interval
samples one every $2^\ell$ elements in the interval and stores them in a local y-fast trie. The y-fast trie of elements $x\in[k\cdot 2^{2^{\ell}}+1,(k+2)2^{2^{\ell}}]$ will store keys $x-k\cdot 2^{2^{\ell}}-1$, and thus will range over a universe of size
$\Oh{2^{2^\ell}}$. Since they store, in total, $\Oh{n/2^\ell}$ elements over a
universe of size $\Oh{2^{2^\ell}}$, the space of all the lpsds adds up to $\Oh{n}$ bits at level $\ell$.

Since a lpsds storing $m_r$ elements uses at most $cm_r$ bits, for some constant $c$, we store them one after the other, reserving $cm_r$ bits for each lpsds storing $m_r$ elements. We store a partial sum data structure $P_\ell$ on the $m_r$ values (if there are less than $2^\ell$ elements, then $m_r=0$ and no lpsds is stored). Then we can find in constant time, using $sum$ on $P_\ell$, the starting point of each lpsds. Structure $P_\ell$ uses at most $\Oh{n}$ bits.

Then, to carry out a predecessor search on interval $[i,j]$, we proceed as 
follows:

\begin{enumerate}
\item We compute $\ell = \lceil \log\log (j-i+1) \rceil$, so that the
query is for sure contained in an (even or odd) interval $k$ of level $\ell$.
Number $k$ is found algebraically in constant time.
\item We use $F_\ell$ to map $k$ to its position $p$ in the nonempty intervals,
and then $B_\ell$ to find the corresponding range $X[i_k..j_k]$. This takes
constant time.
\item If $j_k-i_k+1 \le 2^\ell$, we complete the query with a binary search
on $X[i_k..j_k]$, in time $\Oh{\ell}$, and finish.
\item We use the local predecessor search data structure of interval $p$,
found using $P_\ell$, to determine the subinterval $[i'_k..j'_k] \subseteq [i_k..j_k]$ of size $2^\ell$ where the answer lies. This takes time
$\Oh{\log\log 2^{2^\ell}} = \Oh{\ell}$.
\item We complete the query using binary search on $X[i'_k..j'_k]$, in time
$\Oh{\ell}$.
\end{enumerate}

Our data structures use in total $\Oh{n\log\log U}$ bits for a given level $\ell$, which adds up to $\Oh{n(\log\log U)^2}$ bits in total. They answer queries in time $\Oh{\ell} = \Oh{\log\log(j-i+1)}$ on nonempty intervals. Note that on empty intervals our mmphf $F_\ell$ could return an arbitrary value.

\subsection{A Faster and Smaller Data Structure}

Now we divide $X$ into $\lceil n/b\rceil$ {\em blocks}, each containing $b=\log U$ keys. We build a set $S'$ of {\em sampled} keys by selecting the first element of each block ($X[1],X[b+1],X[2b+1]\ldots$). Thus $S'$ will contain in total $\lceil n/b\rceil$ keys.
Now we use the scheme of Section~\ref{sec:dist-sens}, except for the lpsds implementation, which differs in the choice of the keys it stores. This time we will only have $\log\log (U/n)+1$ levels: we collapse all the levels $\ell$ such that $2^\ell\geq \log (U/n)$. At at each level we will have the same associated mmphf $F_\ell$ and the same partial sums $B_\ell$, but the lpsds are built differently and a smarter encoding yields a space usage of $\Oh{n(\log\log (U/n))^2}$ bits, instead of $\Oh{n(\log\log U)^2}$. The new strategy is to store an lpsds for every interval that contains at least one sampled element (i.e., from $S'$) and to store all the sampled elements in the interval in a fast predecessor search data structure that supports queries in time $\Oh{\log\frac{2^\ell}{\log w}}$. Note that for levels where $2^\ell\geq \log (U/n)$, it is sufficient to achieve time $\Oh{\lg\frac{\log(U/n)}{\log w}}$. All these levels are collapsed, as anticipated, into a singe level where a predecessor search data 
structure is built on the on the $\Oh{n/b}$ sampled keys (each of length $\log U$), which answers queries in time $\Oh{\lg\frac{\log(U/n)}{\log w}}$ and uses $\Oh{(n/b)\cdot\log U}=\Oh{n}$ bits (e.g., \cite{BN12}).

Each of the non-collapsed levels $\ell \le \log\log (U/n)$ stores the same kind of predecessor data structure \cite{BN12}, over $\lceil n/b\rceil$ keys of length $2^\ell$. Each such structure uses $\Oh{(n/b)\cdot 2^\ell} = 
\Oh{(n/b)\lg(U/n)} = \Oh{n}$ bits, and answers queries in time
$\Oh{\lg\frac{2^\ell}{\lg w}}$. Added over all the non-collapsed levels, the
space is $\Oh{n\log\log (U/n)}$ bits. 

Note that those lpsds are built on universes of size $2^{2^\ell} \le U/n$, and
hence the keys require only $\lg (U/n)$ bits.
Since the lpsds are only built on sampled keys, they can only determine a predecessor among the sampled keys. The real answer will be inside the block that separates two sampled keys. To complete the search inside a block of $b$ keys, we store a predecessor data structure \cite[Lem.~3.3]{GRR09} for each block. The structure is an index that uses $\Oh{b\log\log (U/n)}$ bits per block (in addition to a global precomputed table of $\Oh{U^\epsilon}$ bits, any constant $\epsilon>0$) and computes the predecessor in constant time for any $b$ polylogarithmic in
$U$, with $O(1)$ accesses to the data. 
Added over all the levels, the space of these structures 
is $\Oh{n(\log\log (U/n))^2}$ bits.

\subsubsection{Queries}

In a level $\ell\leq \log\log (U/n)$, and given the range $[i,j]$ determined by $F_\ell$ and $B_\ell$, we first construct a range $[i_0,j_0]$, which is the largest subinterval of $[i,j]$ aligned to block boundaries. We first use the
predecessor structure of the block $[j_0,j_0+b]$ to look for a predecessor of $j$. If it
exists, this is the answer. Otherwise, we carry out a query on the interval
$[i_0-b,j_0]$, which cannot be empty if $[i,j]$ is nonempty and is handled with
the proper lpsds in time $\Oh{\lg\frac{2^\ell}{\lg w}} =
\Oh{\lg\frac{\lg(j-i+1)}{\lg w}}$.

A query on the collapsed level, on the other hand, simply uses the predecessor data structure for that level. This gives our result.

\begin{theorem} \label{thm:distpred}
Given $n$ points in the discrete universe $[1,U]$ stored in an array $X$, 
there exists a data structure using $\Oh{n(\lg\lg(U/n))^2}$ bits of space that
and solves in time $\Oh{\lg\frac{\lg(j-i+1)}{\lg w}}$, and with $\Oh{1}$ 
accesses to the array $X$, the following query: Given a range $[i,j]$ known to 
contain some element in $X$, return the predecessor of $j$.
\end{theorem}

\section{Number of Points in a Range}
\label{sec:count}

In this section we describe a {\em one-dimensional range counting data 
structure} that handles $n$ points in $[1,U]$ and can count the number $occ$ 
of points in any range $[i,j]$, faster when the range is shorter and when there
are more points to count. We start with a simple solution that takes time
$\Oh{\lg\lg\frac{j-i+1}{occ+1}}$ and $\Oh{n\lg U}$ bits of space. Then we 
improve upon it to obtain a faster and smaller data structure, which in 
particular requires sublinear space overhead on top of any representation of 
the array.

\subsection{A Simple Data Structure} \label{sec:linear-datastructure}

We use $\lceil \log n\rceil-1$ {\em levels}. At each level $\ell \ge 2$ we build a data structure that efficiently answers queries of length between $2^{\ell-2}+1$ and $2^{\ell-1}$. Our structure defines specific {\em intervals} and {\em subintervals}. For clarity we will refer to {\em ranges} to denote any other range of the universe.

Given a level $\ell$, we divide the universe into $\lceil U/2^{\ell-1}\rceil$ overlapping {\em intervals} of size $2^\ell$, so that interval number $k$ will be $[2^{\ell-1}k+1,2^{\ell-1}k+2^\ell]$. It is clear that any range of size at most $2^{\ell-1}$ will be included in at least one interval.

We only consider nonempty intervals. We can have at most $2n$ nonempty intervals, as each point belongs to 2 intervals. We use a mphf $f_\ell$ that maps the $n'\leq 2n$ nonempty intervals into unique numbers in $[1,n']$. The mphf uses $\Oh{n'+\log\log U}=\Oh{n+\log\log U}$ bits of space and answers queries in constant time \cite{HT01}. It gives a correct answer only if we query it for a 
nonempty interval.

We consider how to solve queries on nonempty intervals. Suppose that an interval $k$ contains $n_k$ elements. We cut the interval into $n_k$ equally-sized {\em subintervals}, of size $\lceil2^\ell/n_k\rceil$ (the last subinterval can be shorter).
We use a prefix sum data structure to store the number of elements in each subinterval of the interval $k$. That prefix sum structure uses $\Oh{n_k}$ bits. The space usage over all the prefix-sum data structures for all the intervals is $\Oh{n}$ bits. We concatenate the memory areas of the prefix-sum data structures of the intervals (in the order given by the mphf) and store another bitmap $D_\ell$ that marks the beginning of the prefix-sum data structure of each interval. This new bitmap also uses $\Oh{n}$ bits.

We store one instance of this data structure for levels 
$\ell=2$ to $\ell=\lceil\log n\rceil$.
Each structure uses $\Oh{n+\log\log U}$ bits of space, resulting in 
$\Oh{n\log n + \log n \log\log U}$ bits overall. In addition, we store one
instance of the predecessor data structure of Section~\ref{sec:pred}, and a
{\em range-emptiness} data structure, which tells in constant time whether a 
range $[i,j]$ contains any point, using $\Oh{n\log U}$ bits \cite{ABR01}. 

\subsubsection{Queries.} \label{sec:linear-queryanswer}

We first perform a range-emptiness query to determine whether the query range $[i,j]$ contains at least one element. If not, we immediately return $0$.
Otherwise we compute $\ell=\lceil\log(j-i+1)\rceil+1$ and algebraically determine the interval of level $\ell$ that encloses $[i,j]$. We answer the query using that interval, which we denote $[I,J]$.

We first use $f_\ell$ to find the index $k$ of the interval $[I,J]$. Because we the interval is nonempty, the mphf gives a meaningful answer. Next we use $D_\ell$ to recover the prefix-sum data structure for the interval $[I,J]$. Then, we find the subinterval $[I_0,J_0]$ of $[I,J]$ that contains $i$ and the subinterval $[I_1,J_1]$ that contains $j$. The number of elements in $[i,j]$ equals the sum of the the number of elements in the three ranges $[i,J_0]$, $[J_0+1,I_1-1]$ and $[I_1,j]$. The count of the range $[J_0+1,I_1-1]$ is found in constant time using the prefix sum structure associated to interval $[I,J]$, as the range $[J_0+1,I_1-1]$ is aligned to subinterval boundaries.

What remains is to determine the counts in the two tail ranges $[i,J_0]$ and 
$[I_1,j]$. We only show how to determine the count in range $[i,J_0]$; the 
other case is symmetric. First we query the range-emptyness data structure to 
determine whether the subinterval $[J_0-\lceil 2^\ell/n_k\rceil+1,i-1]$ is 
empty. If it is, then the count in $[i,J_0]$ is the same as in 
$[J_0-\lceil 2^\ell/n_k\rceil+1,J_0]$, and thus can be computed from the prefix
sum data structure. Otherwise, we can carry out two predecessor queries, using 
the structure of Section~\ref{sec:pred}, for the intervals 
$[J_0-\lceil 2^\ell/n_k\rceil+1,i-1]$ and $[J_0-\lceil 2^\ell/n_k\rceil+1,J_0]$,
knowing that both intervals are nonempty. We count the number of elements in 
$[i,J_0]$ by subtracting the rank of the predecessor of $i-1$ from the rank of 
the predecessor of $J_0$.

The query time is dominated by that of the predecessor search,
$\Oh{\log\frac{\log\lceil2^\ell/n_k\rceil}{\log w}}$ according to
Theorem~\ref{thm:distpred}. Now note that $n_k$, the number of elements in 
$[I,J]$, is at least $occ$. On the other hand $J-I+1=2^\ell$ is at most 
$4(j-i+1)$. We thus conclude that the query time is 
$\Oh{\log\frac{\log\frac{j-i+1}{occ+1}}{\log w}}$, as promised. The space,
however, is still $\Oh{n\lg U}$ bits. Now we introduce an improved solution
that reduces it to $\Oh{n\sqrt{\log (U/n)}}$ bits.

\subsection{A Smaller Data Structure}\label{subsec:compressed_fast_maj}

We now modify the data structure of Section~\ref{sec:linear-datastructure}. 
We only build this structure up to level $\ell_0 =
\sqrt{\log\delta}+\log\log\delta$, where $\delta=\lceil U/n\rceil$, and assume
$\delta\geq\log n$ (the case $\delta<\log n$ will be considered at the end of
the section). That is, we only build data structures to handle
intervals of sizes $2,4,\ldots,2^{\ell_0}=2^{\sqrt{\log\delta}}\log\delta$. The structure now uses $\Oh{n\sqrt{\log\delta}}$ bits of space, since at each level it uses $\Oh{n}$ bits. In Section~\ref{sec:range_emptiness}, we also build a more space-efficient range-emptiness index using $\Oh{n\sqrt{\log\delta}}$ bits (on top of a table of keys in sorted order) and answering range emptiness queries in constant time.

We now describe how the upper levels $\ell > \ell_0$ are handled. For any such upper level we store only intervals that have {\em density} at least $1/2^{\sqrt{\log\delta}}$. As every interval of an upper level 
is of size at least $2^{\sqrt{\log\delta}}\log\delta$, we only store intervals 
that contain at least 
$\frac{2^{\sqrt{\log\delta}}\log\delta}{2^{\sqrt{\log\delta}}}=\log\delta$
elements. More generally, in level $\ell = \ell_0 + \ell'$, stored intervals
contain at least $2^{\ell'}\lg\delta$ elements. Thus, we can store these
intervals in dictionaries $R_\ell$ (instead of weaker mmphfs $F_\ell$), which
use $\Oh{\log(U/n)}=\Oh{\log\delta}$ bits per stored interval and answer to
membership queries in constant time. At level $\ell = \ell_0 + \ell'$, there
are $\Oh{n/\left(2^{\ell'}\lg\delta\right)}$ stored intervals, so dictionary 
$R_\ell$ uses $\Oh{n/2^{\ell'}}$ bits. This adds up to $\Oh{n}$ bits over all 
the upper levels. 

At query time, if we do not find the query interval $[I,J]$ in $R_\ell$, we 
conclude that the interval is of density less than $1/2^{\sqrt{\log\delta}}$. 
Since $[i,j] \subseteq [I,J]$, it follows that $\frac{occ}{j-i+1} <
1/2^{\sqrt{\log\delta}}$, and thus $\lg\lg\delta =
\Oh{\lg\lg\frac{j-i+1}{occ+1}}$. This means that we can answer the counting 
query within the promised time using predecessor searches on $X$, in particular
using the structure of Belazzougui and Navarro~\cite{BN12}, which takes
$\Oh{\log\frac{\log\delta}{\log w}}$ time. To reduce its space, we partition
the universe $U$ into $n$ pieces of length $\delta$, and divide the elements
in each piece into slices of $\lg \delta$ elements. If we have more than one 
slice, we build the predecessor structure \cite{BN12} on the first elements of
the slices. Inside each slice we will use another predecessor structure for
small blocks \cite[Lem~3.3]{GRR09}. Then, upon a predecessor query, an 
$O(n)$-bit partial sums structure leads us to the right piece, the predecessor
structure of the piece leads us to the right slice, and the predecessor 
structure of the slice gives the final predecessor. The time is dominated by
the $\Oh{\log\frac{\log\delta}{\log w}}$ complexity of the predecessor structure
of the piece. As for the space, we have $O(n)$ bits for the partial sums,
$O(n\log \delta / \log\delta) = O(n)$ for the predecessor structures on the
pieces, and $O(n\log\log\delta)$ bits for the predecessor structures on the
slices.

We now describe how the dense upper intervals are handled to answer to
queries using only $\Oh{n\sqrt{\log\delta}}$ bits. For every interval $c_i$ at
level $i$ with density $d$ (where $d$ is rounded to the nearest smaller power
of two) we do not necessarily store the bitmap (the one that stores the
cardinalities of the subintervals of the interval), but instead point (using a
pointer) to $c_j$, which is the interval of highest level $j$ such that
(1) the density (also rounded to the nearest smaller power of two) of $c_j$ is 
at least $d$, and (2) $c_j$ fully encloses $c_i$. 
If $c_i=c_j$ then we store the bitmap of $c_i$. Therefore, at query time, we
simply determine whether the query interval $c_i$ is stored explicitly or has
a pointer to another interval $c_j$. In the second case, we can correctly
solve the query using the data of $c_j$, within the same time complexity (as
it depends only on the rounded density of the interval, which is the same for
$c_i$ and $c_j$). The rest of the section is devoted to analyze the space
usage.

We note that the pointer from $c_i$ to $c_j$ can be encoded using just
$\log\log n+1$ bits: we need only to store the level pointed to,
which requires $\log\log n$ bits, and then we know that only two intervals at
any level $\ell$ can enclose $c_i$, thus the pointer can be uniquely determined
using one additional indicator bit (saying whether the interval is the left or
the right one). Since there are only
$\Oh{n/\left(2^{\ell'}\lg\delta\right)}$ intervals stored at level
$\ell=\ell_0+\ell'$, and
in addition it holds $\delta \ge \lg n$, it follows that there are 
$\Oh{n/\left(2^{\ell'}\lg\lg n\right)}$ stored intervals at upper level
$\ell$, and hence the 
pointers add up to $\Oh{n}$ bits over all the upper levels. 

Now we upper bound the space used by the explicitly stored intervals.
The key issue is to prove that a point appears in at most $2\sqrt{\log\delta}$ bitmaps. To see why, we will first prove that a point appears in at most two bitmaps of a given rounded density $d$. 
In order to prove this let us first prove the following lemma.

\begin{lemma} \label{lem:levels}
There are no three distinct intervals (from any levels) such that the three
pairs of distinct intervals partially overlap each other.
\end{lemma}
\begin{proof}
Assume otherwise. Let the 3 intervals $c_\alpha,c_\beta,c_\gamma$, from levels
$\alpha\le\beta\le\gamma$, respectively. First note that $c_\beta$ overlapping
$c_\alpha$ means that $c_\alpha$ starts at or ends before the middle of
$c_\alpha$. The reason is that $c_\beta$ starts on or ends before multiples of
$2^\beta/2\geq 2^\alpha/2$. Thus in order to overlap with $c_\alpha$ it must
start on or end before a point somewhere strictly inside $c_\alpha$ and the
only point that can be a multiple of $2^\alpha/2$ (and thus possibly a multiple of $2^\beta/2$) is the middle of the interval. 

The same argument holds for $c_\gamma$, which must start on or end before the
middle of $c_\alpha$. We now compare $c_\beta$ and $c_\gamma$. If $c_\beta$
(respectively $c_\gamma$) starts in the middle of $c_\alpha$ and $c_\gamma$
(respectively $c_\beta$) ends before the middle of of $c_\alpha$, then they
are not overlapping. It remains to consider the case that both $c_\beta$ and
$c_\gamma$ start on or end before the middle of $c_\alpha$. In this case,
clearly $c_\beta$ is enclosed in $c_\gamma$, simply because they start or end at the same point and $c_\beta$ is bigger than $c_\gamma$. 
\qed
\end{proof}

From the lemma we can now prove our next goal.

\begin{lemma}
There cannot be a point that participates in three distinct interval
bitmaps with the same rounded density.
\end{lemma}
\begin{proof}
Assume otherwise.
Let a point participate in distinct intervals $c_i,c_j,c_k$ at levels $i\leq
j\leq k$ with the same rounded density $d$. We prove that if this was the case
then the three intervals $c_i,c_j,c_k$ should be partially overlapping, which
is impossible by Lemma~\ref{lem:levels}. First of all, the three intervals
must include the same point, so they must clearly be overlapping. Also, no
interval can be included in the other, as if this was the case then the
included interval would not be explicit but instead point to some of the
intervals that enclose it, as they have the same rounded density. Thus, each
pair of intervals is overlapping and no interval is enclosed in the other, which means that the pairs of intervals are partially overlapping. 
\qed
\end{proof}

As we have exactly $\sqrt{\log\delta}$ distinct levels, we conclude that each
point participates in at most $2\sqrt{\lg\delta}$ explicit bitmaps, and thus
the total space used by all those bitmaps (which store $\Oh{1}$ bits per point
included) is $\Oh{n\sqrt{\log\delta}}$ bits. 

\paragraph{The case $\delta < \log n$.} If $\delta = U/n < \log n$, then
$U < n\log n$. In this case we use a different solution. We split the universe
into $n$ intervals of length $\delta < \log n$. A partial sums data structure
accumulates the number of points in each interval using $O(n)$ bits. Inside
each interval, we store one predecessor data structure \cite[Lem.~3.3]{GRR09},
which will add up to $O(n\log\log\delta)$ bits (plus a global precomputed
table of $O(\delta^\epsilon)$ bits), and will solve predecessor queries in
constant time within the intervals. Then the range counting is easily done in
constant time and using $O(n\log\log(U/n))$ bits.

\subsection{Space-Efficient Range Emptiness}
\label{sec:range_emptiness}

A range-emptiness index that uses $\Oh{n\sqrt{\log U}}$ bits already exists 
\cite{BBPV10}. Its space can be trivially improved to $\Oh{n\sqrt{\log (U/n)}}=
\Oh{\sqrt{\log\delta}}$ bits by dividing $U$ into $n$ intervals of
size $\lceil U/n\rceil$ and storing in a prefix sum data structure the number 
of elements in each interval. Then we build a local range 
emptiness index on the elements that belong to each interval. The index will 
thus use $\Oh{\sqrt{\log (U/n)}}$ bits per element, for a total of 
$\Oh{n\sqrt{\log (U/n)}}$ bits over all the local indexes. Now, given a query 
range, it fully contains zero or more consecutive intervals and partially 
overlaps one or two intervals. The emptiness of the fully contained
intervals is established using the prefix sum structure, while the 
partially overlapped intervals are queried using the local range emptiness 
indexes. Thus a range emptiness query can be decided in constant time, and
our final result is proved.

\begin{theorem} \label{thm:rangecount}
Given $n$ points in the discrete universe $[1,U]$ stored in an array $X[1..n]$, there exists a data structure
using $\Oh{n\sqrt{\log (U/n)}}$ bits of space that returns the number $occ$ of points in any range $[i,j]$ in time $\Oh{\log\frac{\log\frac{j - i + 1}{\occ +1}}{\log w}}$ and in $\Oh{1}$ accesses to the array $X$. The data structure uses precomputed tables that occupy $\Oh{U^\epsilon}$ bits of space (where $0<\epsilon<1$ is any constant), which are independent of the point set.
\end{theorem}

We can slightly adapt this procedure to return some element when the range
is nonempty. The structure used within the intervals \cite{BBPV10} is a weak
prefix search data structure, so it will return some element when it finds 
that the interval is nonempty. In case the only elements are in the sequence
of whole consecutive intervals covered by the partial sum data structure, we
can use the structure to find the first nonempty interval in the sequence
(by searching for the interval where the sum reaches $x+1$, being $x$ the sum
up to the first interval in the sequence, not including it). Once we
have identified a nonempty interval, the weak prefix search data structure
\cite{BBPV10} of this interval will give us one element in it. This feature
will be useful later in the paper.

\section{Number of Points in a Range, Again}
\label{sec:count2}

We propose a different range counting data structure, which performs better
when there are fewer points in the count. Namely, we count in time 
$\Oh{\sqrt{\frac{\log(occ+1)}{\log\log(occ+1)}}}$ using $\Oh{n\sqrt{\lg(U/n)}}$
bits.

We use the finger-search data structure of Andersson and Thorup~\cite{AT07}.
Given $n$ elements in the discrete universe $[1,U]$, it uses $\Oh{n\lg U}$ bits
and answers the following variant of the predecessor query: Given a ``finger''
element $k$ and a query for the predecessor of $x$, it answers in time
$\Oh{\sqrt{\frac{\log(rd(x,k)+1)}{\log\log(rd(x,k)+1)}}}$, where $rd(x,k)$ is 
the number of points lying between $k$ and $x$. In addition, we will use the
range emptiness data structure of Section~\ref{sec:range_emptiness}.

First, we cut the universe into $\lceil U/n\rceil$ intervals of equal size and
store in a prefix sum data structure the number of elements in each interval. 
Given a range counting query, we use the prefix sums to count the number of elements in the intervals that are fully contained in the query range. What remains is to count the number of elements in the up to two intervals that are not fully contained in the query range. 

To that end, we sample one every $b$ keys, with $b=\log (U/n)$, inside each 
interval, and store the sampled keys in a finger-search data structure for that
interval. We do not store the full keys, but only the least significant
$\log (U/n)$ bits, since the $\log n$ upper bits of all the keys inside an 
interval are the same. If an interval contains less than $b$ keys, we do not 
store the finger-search data structure. Overall, the finger-search data 
structures store up to $n/b$ keys, each of $b$ bits, for a total space usage
of $\Oh{n}$ bits. 

The elements between two sampled keys form a {\em block},
and we store one predecessor structure \cite[Lem.~3.3]{GRR09} for each block. These
will add up to $\Oh{n\lg\lg(U/n)}$ bits of space, plus a fixed shared table
of $\Oh{(U/n)^\epsilon}$ bits, for some constant $\epsilon$. The range-emptiness
data structure of Section~\ref{sec:range_emptiness} uses other 
$\Oh{n\sqrt{\log (U/n)}}$ bits, which dominate the overall space. 

Given a range $[i,j]$ fully contained in an interval, we first ask if the range
is empty. If it is, the count is zero. Otherwise, the range emptiness data structure returns some element $k\in[i,j]$. Then we perform two queries on the finger-search data structure, for the points $i$ and $j$, using the finger $k$. Since $rd(i,k),rd(j,k)\leq occ$, the queries take time time at most $\Oh{\sqrt{\frac{\log(occ+1)}{\log\log(occ+1)}}}$, and give us the predecessor $j_0$ of $j$ and the successor $i_0$ of $i$ among the sampled keys stored in the finger-search data structure. Finally, the ranges $[i,i_0]$ and $[j_0,j]$ are contained in blocks, so a predecessor search on each takes constant time using the predecessor structures.

\begin{theorem} \label{thm:rangecount2}
Given $n$ points in the discrete universe $[1,U]$ stored in an array $X[1..n]$, there exists a data structure
using $\Oh{n\sqrt{\log (U/n)}}$ bits of space that returns the number $occ$ of points in any range $[i,j]$ in time $\Oh{\sqrt{\frac{\log(occ+1)}{\log\log(occ+1)}}}$ and in $\Oh{1}$ accesses to the array $X$. The data structure uses precomputed tables that occupy $\Oh{(U/n)^\epsilon}$ bits of space (where $0<\epsilon<1$ is any constant), which are independent of the point set.
\end{theorem}

\section{Counting Elements in a Range}

We now switch to another scenario, where instead of points in a universe we
have a sequence $S[1..n]$ of elements over a discrete alphabet of symbols in
$[1,\sigma]$.
We use the results of the previous sections to answer queries on $S$, on top
of a representation of $S$ that can answer access, rank and select queries.
In this section we show how to count the number of occurrences of a given 
symbol in an array interval, in time that improves with its frequency in
the interval, and using compressed space.

The basic idea is to create, for each symbol $1 \le c \le \sigma$, a point set
$P_c = \{ i,~S[i]=c\}$ over universe $[1,n]$, and reduce the counting for 
symbol $c$ to range counting on $P_c$. We call $n_c$ the number of occurrences 
of $c$ in $S$, and $\delta_c = n/n_c$ its inverse relative frequency. We will 
use the range counting structure of Theorem~\ref{thm:rangecount} for each 
$P_c$, using $\delta=\delta_c$.

We represent $S$ using alphabet partitioning \ignore{(OJO must explain this in more depth in Sec 2; this explanation is right after Thm 10)}, which distributes the alphabet into subalphabets according to the
value of $\lceil \lg \delta_c \cdot \lg n \rceil$. Thus any pair of symbols 
$c$ and $c'$ belonging to the same subalphabet satisfy $\delta_c = \Theta(\delta_{c'})$. In the alphabet partitioned representation, the subalphabets of
polylogarithmic size are represented so that access, rank and select take
constant time, and thus we can solve the counting for those symbols in constant
time using rank queries. Note that if $\delta_c\leq\log n$, it follows that
$n_c \ge n/\lg n$, and thus there cannot be more than $\lg n$ symbols where 
that holds. Therefore all those subalphabets are of logarithmic size, and
we can focus only on the case $\delta_c>\log n$. 

Note that, on those symbols with larger $\delta_c$, rank queries on the 
partitioned representation take time $\Oh{\log\frac{\log\delta_c}{\log w}}$,
and those can be used to solve the counting query. This time is good enough
for intervals of density below $1/2^{\sqrt{\log\delta_c}}$, so this replaces
the predecessor data structure used in Section~\ref{sec:count}.
We also note that the range emptiness data structure needs to access the
array of ``points'' $P_c$. This is simulated with select operations on $S$.

As for the space when $\delta_c>\log n$, the data structure of
Theorem~\ref{thm:rangecount} uses $\Oh{n_c\sqrt{\lg\delta_c}}$ bits. 
We let $nH=n'H'+n''H''$, where $n'H'=\sum n_c\log(\delta_c)$  for all $c$ with 
$\delta_c>\log n$ is the contribution of characters $c$ with $\delta_c>\log n$ 
to the total entropy of the sequence $S$
and $n'$ is the total number of occurrences of such characters ($n''$ and $H''$
are similarly defined for the remaining symbols). It is easy to see 
that $H'\geq \log\log n$. 

By convexity of the logarithm, the total space used for all $c$ with  $\delta_c>\log n$ adds up to 
$\Oh{n' \sqrt{H'}}\leq O(n'\frac{H'}{\sqrt{\log\log n}})\leq O(nH/\sqrt{\log\log n})=o(nH)$ bits. 

We also build the precomputed tables for $\log^2 n$ different values of 
$\delta$, thus
all the precomputed tables for the predecessor structures \cite[Lem.~3.3]{GRR09} occupy $\Oh{n^\epsilon\log^2 n}=o(n)$ bits of space.

Therefore, the space is dominated by the (alphabet partitioned) representation 
of $S$. Apart from operation rank, it supports select and access in constant 
time if we let it use $(1+\epsilon)nH$ bits of space.

\begin{theorem} \label{thm:counting}
For any positive constant $\epsilon$, we can store a sequence $S[1..n]$ over
alphabet $[1,\sigma]$ and with per-symbol entropy $H$, 
within \((1 + \epsilon) n H+o(n)\) bits, such that it
supports operations access and select in $\Oh{1}$ time and rank in 
time $\Oh{\lg\lg\sigma}$. Moreover, given endpoints $i$ and $j$ and a symbol
$c \in [1,\sigma]$, it computes \(occ=\occ (a, S [i..j])\) in time
$\Oh{\log \frac{\log \frac{j - i + 1}{occ+1}}{\log w}}$.
\end{theorem}

\ignore{OJO I changed the lglg n in the denominator to lg w, is this somehow wrong?
If not, this improves Thm 12.

OJO it also said that one can use Thm 3 and obtain another result; I have not
developed this.}

\section{Parameterized Range Minority} \label{sec:minority}

Recall from Section~\ref{sec:intro} that a $\tau$-minority for a range is a distinct element that occurs in that range but is not one of its $\tau$-majorities.  The problem of parameterized range minority is to preprocess a string such that later, given the endpoints of a range and $\tau$, we can quickly return a $\tau$-minority for that range if one exists.  Chan et al. gave a linear-space solution with $\Oh{1 / \tau}$ query time even for the case of variable $\tau$.  They first build a list of \(\lfloor 1 / \tau \rfloor + 1\) distinct elements that occur in the given range (or as many as there are, if fewer) and then check those elements' frequencies to see which are $\tau$-minorities.  There cannot be more than \(\lfloor 1 / \tau \rfloor\) $\tau$-majorities so, if there exists a $\tau$-minority for that range, then at least one must be in the list.  In this section we show how to implement this idea using compressed space.

To support parameterized range minority on \(S [1..n]\) in $\Oh{1 / \tau}$ time, we store data structures supporting $\Oh{1}$-time access, select and partial rank queries on $S$ and a data structure supporting $\Oh{1}$-time range-minimum queries on $C$.  For any positive constant $\epsilon$, we can store these data structures in a total of \((1 + \epsilon) n H + \Oh{n}\) bits.  Given $\tau$ and endpoints $i$ and $j$, in $\Oh{1 / \tau}$ time we use Muthukrishnan's algorithm to build a list of \(\lfloor 1 / \tau \rfloor + 1\) distinct elements that occur in \(S [i..j]\) (or as many as there are, if fewer) and the positions of their leftmost occurrences therein.  We check whether these distinct elements are $\tau$-minorities using the following lemma:

\begin{lemma} \label{lem:check}
Suppose we know the position of the leftmost occurrence of a distinct element in a range.  We can check whether that distinct element is a $\tau$-minority or a $\tau$-majority using a partial rank query and a select query on $S$.
\end{lemma}

\begin{proof}
Let $k$ be the position of the first occurrence of $a$ in \(S [i..j]\).  If \(S [k]\) is the $r$th occurrence of $a$ in $S$, then $a$ is a $\tau$-minority for \(S [i..j]\) if and only if the \((r + \lceil \tau (j - i + 1) \rceil - 1)\)th occurrence of $a$ in $S$ is strictly after \(S [j]\); otherwise $a$ is a $\tau$-majority.  That is, we can check whether $a$ is a $\tau$-minority for \(S [i..j]\) by checking whether
\[S.\select_a \left( \rule{0ex}{2.5ex} S.\rank_a (k) + \lceil \tau (j - i + 1) \rceil - 1 \right) > j\,;\]
since \(S [k] = a\), computing \(S.\rank_a (k)\) is only a partial rank query.
\qed
\end{proof}

To avoid storing $C$, which is used in Muthukrishnan's algorithm, we use
Sadakane's variant \cite{Sad07}, which marks the values found in a bitmap of
size $\sigma$, and stops the recursion when the new symbol to consider is
already marked (for this to work he must first process the left and then the
right interval of the minimum).

This gives us the following theorem, which improves Chan et al.'s solution to use nearly optimally compressed space with no slowdown.

\begin{theorem} \label{thm:fast min}
For any positive constant $\epsilon$, we can store $S$ in \((1 + \epsilon) n H + \Oh{n}\) bits such that later, given the endpoints of a range and $\tau$, we can return a $\tau$-minority for that range (if one exists) in $\Oh{1 / \tau}$ time.
\end{theorem}

Alternatively, for any function \(f (n) = \omega (1)\), we can store our data structures for access, select and partial rank on $S$ and range-minimum queries on $C$ in a total of \(n H + \Oh{n} + o (n H)\) at the cost of select queries taking $\Oh{f (n)}$ time.

\begin{theorem} \label{thm:small min}
For any function \(f (n) = \omega (1)\), we can store $S$ in \(n H + \Oh{n} + o (n H)\) bits such that later, given the endpoints of a range and $\tau$, we can return a $\tau$-minority for that range (if one exists) in $\Oh{(1 / \tau)\,f (n)}$ time.
\end{theorem}

To reduce the space bound of this t to \(n H + o (n (H + 1))\) bits, improving Chan et al.'s solution to use optimally compressed space with nearly no slowdown, we must reduce the space of the range-minority data structure to $o(n)$.

We do this via sparsification. We cut the sequence into blocks of length $f(n)$, choose the $n/f(n)$ minimum values of each block, and build the RMQ data structure on the new array $C'[1..n/f(n)]$. This requires $\Oh{n/f(n)} = o(n)$ bits.
Muthukrishnan's algorithm is then run over $C'$ as follows. We find the minimum
position in $C'$, then recursively process its left interval, then process the
minimum of $C'$ by considering the $f(n)$ corresponding cells in $C$, and 
finally process the right part of the interval. The recursion stops when the
interval becomes empty or when all the $f(n)$ elements in the block of $C$
are marked. In addition we must sequentially process the $2f(n)$ cells of $S$
that only partially overlap blocks in $C'$.
We note that a similar technique is proposed by Hon et 
al.~\cite{HSV09}, but it lacks sufficient detail to ensure correctness. We 
prove such correctness next.

\begin{lemma}
The procedure described correctly identifies all the distinct points in
$S[i..j]$, working over at most $f(n)$ cells per new element discovered.
\end{lemma}
\begin{proof}
We show by induction on the size of the current subinterval $[\ell..r]$ that, 
if we start the procedure with the elements that already appear in $[i..\ell-1]$
marked, then we find and mark the leftmost occurrence of each distinct symbol 
not yet marked, spotting at least one new element per block scanned. 

This is trivial for the empty interval. Now consider
the minimum position in $C'[k']$, which contains the leftmost occurrence of some
element $S[k]$, for some $k$ within the block of $C'[k']$. If $S[k]$ is already
marked, it means it appears in $[i..\ell-1]$, and thus the leftward pointer
$C[k] \ge i$, and so holds for all the values in $C[\ell..r]$. Thus if all the
elements in the block of $C'[k']$ are marked, we can safely stop the procedure.

Otherwise,
before doing any marking, we recursively process the interval to the left of
block $k'$, which by inductive hypothesis marks the unique elements in that
interval. Now we process the current block of size $f(n)$, finding at least
the new occurrence of element $S[k]$ (which cannot appear to the left of $k'$).
Once we mark the new elements of the current block, we process the interval
to the right of $k'$, where the inductive hypothesis again holds. \qed
\end{proof}

By using this procedure to obtain any $\lfloor 1/\tau\rfloor+1$ 
distinct elements, we obtain the improved result.

\begin{theorem} \label{thm:small min better}
For any function \(f (n) = \omega (1)\), we can store $S$ in \(n H + o(n)(H+1))\) bits such that later, given the endpoints of a range and $\tau$, we can return a $\tau$-minority for that range (if one exists) in $\Oh{(1 / \tau)\,f (n)}$ time.
\end{theorem}

\section{Parameterized Range Majority with Fixed $\tau$} \label{sec:fixed maj}

The standard approach to finding $\tau$-majorities, going back to Misra and Gries' work, is to build a list of $\Oh{1 / \tau}$ candidate elements and then verify them.  For parameterized range majority, an obvious way to verify candidates is to use rank queries.  The problem with this approach is that, as noted in Section~\ref{subsec:queries}, we cannot support general rank queries in \(o (\log (\log \sigma / \log w))\) time while using \(n \log^{\mathcal{O} (1)} n\) space; e.g., with only linear space, we cannot support general rank queries in $\Oh{1}$ time when the alphabet is super-polylogarithmic.  If we can find the position of candidates' first occurrences in the range, however, then by Lemma~\ref{lem:check} we can check them using only partial rank and select queries.

Suppose we want to support parameterized range majority on \(S [1..n]\) for a fixed threshold $\tau$.  We first store data structures that support access, select and partial rank on $S$ in $\Oh{1}$ time, which takes $\Oh{n}$ space.  For \(0 \leq b \leq \lfloor \log n \rfloor\), let \(F_b [1..n]\) be the binary string in which \(F_b [k] = 1\) if the distinct element \(S [k]\) occurs at least \(\tau 2^b\) times in \(S [k..k + 2^{b + 1} - 1]\); and let $S_b$ and $C_b$ be the subsequences of $S$ and $C$, respectively, consisting of those elements flagged by 1s in $F_b$.  We store $F_b$ in $\Oh{n}$ bits such that we can support access, rank and select queries on $F_b$ in $\Oh{1}$ time.  Notice we can implement an access query on $S_b$ or $C_b$ as a select query on $F_b$ and access queries on $S$ or $C$, respectively.  As described in Section~\ref{subsec:listing}, we can implement an access query to $C$ as access, select and partial rank queries on $S$.  We also store an $\Oh{1}$-time range-minimum data
structure for $C_b$, which takes $\Oh{|S_b|}$ bits.

With these data structures, given endpoints $i$ and $j$ with \(\lfloor \log (j
- i + 1) \rfloor = b\), we use Muthukrishnan's algorithm to list the distinct
elements in \(S_b [F_b.\rank_1 (i)..\)
\(F_b.\rank_1 (j)]\) and the positions of their leftmost occurrences therein; we then use select queries on $F_b$ to find the positions of those elements in $S$.  That is, we list the distinct elements in \(S [i..j]\) that are flagged by 1s in $F_b$ and the positions of their leftmost flagged occurrences therein.  We then apply Lemma~\ref{lem:check} to each of these elements, treating the positions of their leftmost flagged occurrences as the positions of their leftmost occurrences.  Since each distinct element in \(S [i..j]\) that is flagged in $F_b$ occurs at least \(\tau 2^b\) times in \(S [i..j + 2^{b + 1} - 1] \subset S [i..i + 2^{b + 2}]\), there are $\Oh{1 / \tau}$ of them and we use a total of $\Oh{1 / \tau}$ time.

Notice that the leftmost flagged occurrences of a distinct element $a$ in \(S [i..j]\) may not necessarily be the leftmost occurrence therein.  However, if $a$ is a $\tau$-majority in \(S [i..j]\) then, by definition, $a$ occurs at least \(\tau (j - i + 1) \geq \tau 2^b\) times in \(S [i..j] \subset S [i..i + 2^{b + 1} - 1]\), so $a$'s leftmost occurrence in \(S [i..j]\) is flagged by a 1 in $F_b$ and, therefore, we apply Lemma~\ref{lem:check} to it.  It follows that we return each $\tau$-majority in \(S [i..j]\).

We store only one set of data structures supporting access, select and partial rank on $S$.  Summing over $b$ from 0 to \(\lfloor \log n \rfloor\), the data structures for range-minimum queries take a total of $\Oh{n \log n}$ bits, which is $\Oh{n}$ words.  Therefore, we have the first linear-space data structure with worst-case optimal $\Oh{1 / \tau}$ query time for Karpinski and Nekrich's original problem of parameterized range majority with fixed $\tau$.

\begin{theorem} \label{thm:fixed maj}
Given a threshold $\tau$, we can store a string in linear space and support parameterized range majority in $\Oh{1 / \tau}$ time.
\end{theorem}

\section{Parameterized Range Majority with Variable $\tau$} \label{sec:variable maj}

\subsection{Nearly linear space with optimal query time} \label{subsec:fast maj}

Suppose we have an instance of the data structure from Theorem~\ref{thm:fixed maj} for each threshold \(1, 1 / 2, 1 / 4, \ldots, 1 / 2^{\lceil \log n \rceil}\), which takes a total of $\Oh{n \log n}$ space.  Given endpoints $i$ and $j$ and a threshold $\tau$, we can use the instance for threshold \(1 / 2^{\lceil \log (1 / \tau) \rceil}\) to build a list of $\Oh{1 / \tau}$ candidate elements and then check them with Lemma~\ref{lem:check}; this takes a total of $\Oh{1 / \tau}$ time and returns all the $\tau$-majorities in \(S [i..j]\).  Gagie et al. used a variant of this idea to obtain the first data structure for variable $\tau$.  We can easily reduce our space bound to $\Oh{n \log \sigma}$ because, if \(1 / \tau \geq \sigma\), then we can simply use Muthukrishnan's algorithm with $S$ and $C$ to list in \(\Oh{\sigma} = \Oh{1 / \tau}\) time all the distinct elements in \(S [i..j]\) and the positions of their leftmost occurrences therein, then check them with Lemma~\ref{lem:check}.

Notice that we need store only one set of data structures supporting access, select and partial rank on $S$.  Also, if \(S [k]\) is a \((1 / 2^t)\)-majority in a range, then it is also a \((1 / 2^{t'})\)-majority for all \(t' \geq t\).  It follows that if, instead of querying only the instance for the threshold \(1 / 2^{\lceil \log (1 / \tau) \rceil}\), we query the instances for all the thresholds \(1, 1 / 2, 1 / 4, \ldots, 1 / 2^{\lceil \log (1 / \tau) \rceil}\) --- which still takes \(\Oh{\sum_{t = 0}^{2^{\lceil \log (1 / \tau) \rceil}} 2^t} = \Oh{1 / \tau}\) time --- then we can modify the instances to reduce the total number of 1s in their binary strings.  Specifically, for \(0 \leq t \leq \lceil \log \sigma \rceil\), let $F_b^t$ be the binary string $F_b$ in the instance for threshold \(1 / 2^t\); we modify $F_b^t$ such that \(F_b^t [k] = 1\) if and only if the number of occurrences of the distinct element \(S [k]\) in \(S [k..k + 2^{b + 1} - 1]\) is at least $2^{b - t}$ times but less than $2^{b - t + 
1}
$.

For \(0 \leq b \leq \lfloor \log n \rfloor\) and \(1 \leq k \leq n\), we have
\(F_b^t [k] = 1\) for at most one value of $t$.  Therefore, all the binary
strings contain a total of at most \(n (\lfloor \log n \rfloor + 1)\) copies of 1,
so all the range-minimum data structures take a total of $\Oh{n \log n}$ bits.
Since the binary strings have total length \(n \lceil \log n \rceil \lceil
\log \sigma \rceil\), we can use P\v{a}tra\c{s}cu's data structure to store
them in a total of $\Oh{n \log (n) \log \log \sigma}$ bits.  A slightly neater
approach is to represent all the binary strings \(F_b^0, \ldots, F_b^{\lceil
\log \sigma \rceil}\) as a single string \(T_b [1..n]\) in which \(T_b [k] =
t\) if \(F_b^t [k] = 1\), and $\infty$ if there is no such value $t$.  We can
implement access, rank and select queries on \(F_b^0, \ldots, F_b^{\lceil \log
\sigma \rceil}\) by access, rank and select queries on $T_b$.  Since $T_b$ is
an alphabet of size $\Oh{\log \sigma}$, we can use Ferragina et al.'s data
structure to store it in $\Oh{n \log \log \sigma}$ bits and support access, rank and select queries in $\Oh{1}$ time.  Either way, in total we use $\Oh{n \log \log \sigma}$ space.

\begin{theorem} \label{thm:fast maj}
We can store $S$ in $\Oh{n \log \log \sigma}$ space such that later, given the endpoints of a range and $\tau$, we can return the $\tau$-majorities for that range in $\Oh{1 / \tau}$ time.
\end{theorem}

\subsection{Optimally compressed space with nearly optimal query time} \label{subsec:small maj}

To be able to apply Lemma~\ref{lem:check}, we must be able to find the leftmost occurrence of each $\tau$-majority in a range.  For this reason, we may flag many occurrences of the same distinct element even when they appear in close succession, because we cannot know in advance where the query range will start.  As discussed in Section~\ref{sec:fixed maj}, however, if we have a data structure that supports rank queries on $S$, then it is sufficient for us to build a list of $\Oh{1 / \tau}$ candidate elements that includes all the $\tau$-majorities --- without any information about positions --- and then verify them using rank queries.  This lets us flag fewer elements and so reduce our space bound, at the cost of using slightly suboptimal query time.

We store an instance of Barbay et al.'s data structure \cite{BCGNN13} supporting access on $S$ in $\Oh{1}$ time and rank and select on $S$ in $\Oh{\log \log \sigma}$ time, which takes \(n H + o (n (H + 1))\) bits.  For \(0 \leq t \leq \lceil \log \sigma \rceil\) and \(\lfloor \log (2^t \log \log \sigma) \rfloor \leq b \leq \lfloor \log n \rfloor\), we divide $S$ into blocks of length $2^{b - 1}$ and store data structures supporting access, rank and select on the binary string \(G_b^t [1..n]\) in which \(G_b^t [k] = 1\) if, first, the distinct element \(S [k]\) occurs at least $2^{b - t}$ times in \(S [k - 2^{b + 1}..k + 2^{b + 1}]\) and, second, \(S [k]\) is the leftmost or rightmost occurrence of that distinct element in its block.  We also store an $\Oh{1}$-time range-minimum data structure for the subsequence of $C$ consisting of elements flagged by 1s in $G_b^t$.

The number of distinct elements that occur at least $2^{b - t}$ times in a range of size $\Oh{2^b}$ is $\Oh{2^t}$, so there are $\Oh{2^t}$ elements in each block flagged by 1s in $G_b^t$.  It follows that we can store an instance of P\v{a}tra\c{s}cu's data structure supporting $\Oh{1}$-time access, rank and select on $G_b^t$ in $\Oh{n 2^{t - b} (b - t) + n / \log^3 n}$ bits in total; we need $\Oh{n2^{t-b}}$ bits for the corresponding range-minimum data structure.  Summing over $t$ from 0 to \(\lceil \log \sigma \rceil\) and over $b$ from \(\lfloor \log (2^t \log \log \sigma) \rfloor\) to \(\lfloor \log n \rfloor\), calculation shows we use a total of \(\Oh{\frac{n \log \sigma \log \log \log \sigma}{\log \log \sigma} + \frac{n}{\log n}} = o (n \log \sigma)\) bits for the binary strings and range-minimum data structures.  Therefore, including the instance of Barbay et al.'s data structure for $S$, we use \(n H + o (n \log \sigma)\) bits altogether.

Given endpoints $i$ and $j$ and a threshold $\tau$, if
\[\lfloor \log (j - i + 1) \rfloor
< \left\lfloor \log \left( 2^{\lceil \log (1 / \tau) \rceil} \log \log \sigma \right) \right\rfloor\,,\]
then we simply run Misra and Gries' algorithm on \(S [i..j]\) in \(\Oh{j - i} = \Oh{(1 / \tau) \log \log \sigma}\) time.  Otherwise, we use Muthukrishnan's algorithm to list the distinct elements flagged by 1s in $G_b^t$, where \(t = \lceil \log (1 / \tau) \rceil\) and \(b = \lfloor \log (j - i + 1) \rfloor \geq \lfloor \log (2^t \log \log \sigma) \rfloor\), and use rank queries on $S$ to check whether each of them is a $\tau$-majority in \(S [i..j]\).  Since \(S [i..j]\) overlaps at most 5 blocks of length $2^{b - 1}$, it contains $\Oh{1 / \tau}$ distinct elements flagged by 1s in $G_b^t$; therefore, Muthukrishnan's algorithm takes $\Oh{1 / \tau}$ time and we use a total of $\Oh{(1 / \tau) \log \log \sigma}$ time for all the rank queries on $S$.

Since \(S [i..j]\) cannot be completely contained in a block of length $2^{b - 1}$, if \(S [i..j]\) overlaps a block then it includes one of that block's endpoints.  Therefore, if \(S [i..j]\) contains an occurrence of a distinct element $a$, then it includes the leftmost or rightmost occurrence of $a$ in some block.  Suppose $a$ is a $\tau$-majority in \(S [i..j]\).  For all \(i \leq k \leq j\), $a$ occurs at least $\tau 2^b \ge 2^{b - t}$ times in \(S [k - 2^{b + 1}..k + 2^{b + 1}]\), so some occurrence of $a$ in \(S [i..j]\) is flagged by a 1 in $G_b^t$.  Therefore, we return $a$.

\begin{theorem} \label{thm:small maj}
We can store $S$ in \(n H + o (n \log \sigma)\) bits such that later, given the endpoints of a range and $\tau$, we can return the $\tau$-majorities for that range in $\Oh{(1 / \tau) \log \log \sigma}$ time.
\end{theorem}

In order to reduce the space further, we open the black-box of Barbay et al.'s
data structure. This separates the sequence symbols into $\lg^2 n$ classes 
according to their frequencies. A sequence $K[1,n]$, where $K[i]$ is the class
to which $S[i]$ is assigned, is represented using a (multi-ary) wavelet tree 
\cite{FMMN07}, which supports constant-time access, rank, and select, since
the alphabet of $K$ is of polylogarithmic size. For each class $c$, a sequence
$S_c[1..n_c]$ contains the subsequence of $S$ of the symbols $S[i]$ where
$K[i]=c$ (note that if $S[i]=S[j]$ then $K[i]=K[j]$). They represent $K$, and
the subsequences $S_c$ where $\sigma_c \le \lg n$, using 
wavelet trees. The subsequences $S_c$ over larger alphabets are represented
using Golynski et al.'s structure \cite{GMR06}. The wavelet tree for $K$
takes $nH_K+o(n)$ bits, where $H_K$ is the entropy of $K$, the wavelet trees
for the strings $S_c$ take $n_c\lg\sigma_c + o(n_c)$ bits, and Golynski et
al.'s structures take $n_c\lg\sigma_c + o(n_c\lg\sigma)$ bits. 
Barbay et al.\ show that these spaces
add up to $nH + o(n)(H+1)$ and that one can support access, rank and select on 
$S$ via access, rank and select on $K$ and some $S_c$.

We will solve a $\tau$-majority query on $S[i..j]$ as 
follows. We first run a $\tau$-majority query on string $K$. This will yield
the at most $1/\tau$ classes of symbols that, together, occur at least 
$\tau(j-i+1)$ times in $S[i..j]$. The classes excluded from this result cannot
contain symbols that are $\tau$-majorities. Now, for each included class $c$,
we map the interval $S[i..j]$ to $S_c[i_c..j_c]$ in the subsequence of its
class, since $i_c = K.\rank_c(i-1)+1$ and $j_c = K.\rank_c(j)$, and then 
run a $\tau_c$-majority query on $S_c[i_c..j_c]$, for 
$\tau_c = \tau(j-i+1)/(j_c-i_c+1)$. The results obtained for each 
considered class $c$ are reported as $\tau$-majorities in $S[i..j]$.

To run the $\tau_c$-majority queries on the sequences $S_c$ that are 
implemented with Golynski et al.'s structure, we store our representation of 
Theorem~\ref{thm:small maj}. This will add $o(n_c\log\sigma_c)$ bits, which 
does not change the asymptotic space of the data structure. Therefore we will
take $\Oh{(1/\tau_c)\log\log \sigma}$ time to solve those majority queries.
Added over all the possible $\tau_c$ values, we have
$\sum_c (1/\tau_c) \Oh{\log\log\sigma} =
\sum_c (j_c-i_c+1)/(\tau(j-i+1)) \Oh{\log\log\sigma} =
\Oh{(1/\tau)\lg\lg\sigma}$ total time on those sequences.

Let us now consider the case of the query on $K[i..j]$. 
Since the alphabet size is $\log^2 n$, we will partition it
into $\lg^{2/3} n$ classes of $\lg^{4/3} n$ consecutive symbols, and 
subpartition these new classes into $\lg^{2/3} n$ classes of $\lg^{2/3} n$ 
symbols. This works just like the general partitioning into classes: we
perform a $\tau$-majority query in the first level, then several queries
adding up to cost $1/\tau$ on the second level, and then several queries
adding up to cost $1/\tau$ on the third level, and then go to the subsequences
$S_c$. It is sufficient to show that we can perform a $\tau$-majority query on
any sequence with alphabet size $\lg^{2/3} n$ to obtain the result. The
entropies of the three sequences add up to $nH_K+o(n)$ (indeed, this leveled
partitioning is how the wavelet tree is actually organized).

To solve a $\tau$-majority query on a sequence with alphabet size 
$\sigma' = \lg^{2/3} n$, we will use again Theorem~\ref{thm:small maj}, with a 
slightly larger block size, 
$\lfloor \lg(2^t \lg\lg\sigma \lg n / \lg\lg n) \rfloor$ $\le b \le \lfloor \lg n 
\rfloor$, and for $0 \le t \le \lceil \lg\sigma' \rceil$. 
Thus the structures $G_b^t$ and 
the range-minimum data structures add up to 
$\Oh{\frac{n\lg\sigma' (\lg\lg n)^2}{\lg\lg\sigma \lg n}} =
\Oh{\frac{n(\lg\lg n)^3}{\lg n}} = o(n)$.

Since the wavelet tree implements rank, select and access in constant time,
the $\tau$-majority operation is solved in time $\Oh{1/\tau}$, except on 
blocks of size $b_0 = \Oh{(1/\tau) \lg\lg\sigma \lg n / \lg\lg n}$, which have 
to be solved sequentially in time $\Oh{(1/\tau)\lg\lg\sigma}$. As before, we
can find the majorities in time $O(\sigma')$ by using rank over all the
symbols, so if $1/\tau \ge \sigma'$ we can simply do this to 
achieve $\Oh{1/\tau}$ time. Otherwise, we can maintain counters for
all the $\sigma'$ distinct symbols, each using $\Oh{\lg\lg n}$ bits to
distinguish values from $0$ to $b_0 = \Oh{\sigma'\lg n}$, using $o(\lg n)$
bits overall. Therefore a universal table lets us read chunks of
$(\lg_{\sigma'} n)/2 = \Theta(\lg n / \lg\lg n)$ symbols and increase the
corresponding counters in constant time. Thus the block can be
processed sequentially in time $\Oh{(1/\tau) \lg\lg\sigma}$.

Finally, the same technique used for string $K$ can be used for the sequences
$S_c$ that are represented with wavelet trees, since their alphabet size is
just $\lg n$. Overall, we have managed to reduce the redundancy of our
representation.

\begin{theorem} \label{thm:small maj better}
We can store $S$ in \(n H + o(n)(H+1)\) bits such that later, given the endpoints of a range and $\tau$, we can return the $\tau$-majorities for that range in $\Oh{(1 / \tau) \log \log \sigma}$ time.
\end{theorem}

Since our solution includes an instance of Barbay et al.'s data structure, we can also support $\Oh{1}$-time access to $S$ and $\Oh{\log \log \sigma}$-time rank and select.

\subsection{Faster query time with nearly optimally compressed space} \label{subsec:linear maj}

Recall from Section~\ref{subsec:fast maj} that, if \(1 / \tau \geq \sigma\), then we can simply use Muthukrishnan's algorithm to list all the distinct elements in a range and then check them with Lemma~\ref{lem:check}; therefore, we can assume \(1 / \tau < \sigma\).  In this subsection we use our new data structure with density-sensitive query time for one-dimensional range counting of Theorem~\ref{thm:counting} to obtain a nearly optimally compressed data structure for parameterized range majority with $\Oh{(1 / \tau) \log \frac{\log (1 / \tau)}{\log w}}$ query time.

To obtain a compressed data structure for parameterized range majority with $\Oh{(1 / \tau) \log \frac{\log (1 / \tau)}{\log w}}$ query time, we combine our solution from Theorem~\ref{thm:small maj better} with Theorem~\ref{thm:counting}.  Instead of using $\Oh{\log \log \sigma}$-time rank queries to check each of the $\Oh{1 / \tau}$ candidate elements returned by Muthukrishnan's algorithm, we use range-counting queries.  We can make all $\Oh{1 / \tau}$ range-counting queries each take $\Oh{\log \frac{\log (1 / \tau)}{\log w}}$ time because, if one starts taking too much time, then the distinct element we are checking cannot be a $\tau$-majority and we can stop the query early.  (In fact, as we will show in the full version of this paper, our data structure from Theorem~\ref{thm:counting} does not need such intervention.)
This gives us our final result:

\begin{theorem} \label{thm:sensitive maj}
We can store $S$ in \((1 + \epsilon) n H + o(n)\) bits such that later, given the endpoints of a range and $\tau$, we can return the $\tau$-majorities for that range in $\Oh{(1 / \tau) \log \frac{\log (1 / \tau)}{\log w}}$ time.
\end{theorem}

Notice our solution in Theorem~\ref{thm:sensitive maj} takes optimal $\Oh{1 / \tau}$ time when \(1 / \tau = \log^{\mathcal{O} (1)} n\).  Again, we can also support access and select in $\Oh{1}$ time and rank in $\Oh{\log \log \sigma}$ time.

\section{Frequent Range Modes}

We note that we can use our data structures from Theorem~\ref{thm:small maj better} to find a range mode quickly when it is actually reasonably frequent.  Suppose we want to find the mode $x$ of \(S [i..j]\).  To do this, we perform multiple range $\tau$-majority queries on \(S [i..j]\), starting with \(\tau = 1\) and repeatedly reducing it by a factor of 2 until we find at least one $\tau$-majority.  This takes
\[\Oh{\left(1 + 2 + 4 + \ldots + 2^{\left\lceil \log \frac{j - i + 1}{\occ (x, S [i..j])} \right\rceil} \right) \log \log \sigma}
= \Oh{\frac{(j - i + 1) \log \log \sigma}{\occ (x, S [i..j])}}\]
time and returns a list of the $\Oh{\frac{j - i + 1}{\occ (x, S [i..j])}}$ elements that occur at least \(\frac{j - i + 1}{2^{\left\lceil \log \frac{j - i + 1}{\occ (x, S [i..j])} \right\rceil}}\) times in \(S [i..j]\).  We use rank queries to determine which of these elements is the mode $x$, again in \Oh{\frac{(j - i + 1) \log \log \sigma}{\occ (x, S [i..j])}} time.

\begin{theorem} \label{thm:frequent modes}
We can store $S$ in \(n H + o (n) (H + 1)\) bits such that later, given endpoints $i$ and $j$, we can return the mode $x$ of \(S [i..j]\) in $\Oh{\frac{(j - i + 1) \log \log \sigma}{\occ (x, S [i..j])}}$ time.
\end{theorem}

\section{Conclusions} \label{sec:conclusions}

We have given the first linear-space data structure for parameterized range majority with query time $\Oh{1 / \tau}$, which is worst-case optimal in terms of $n$ and $\tau$.  Moreover, we have improved the space bounds for parameterized range majority and minority in the important case of variable $\tau$.  For parameterized range majority with variable $\tau$, we have achieved nearly linear space and worst-case optimal query time, or compressed space with a slight slowdown.  For parameterized range minority, we have improved Chan et al.'s solution to use nearly compressed space with no slowdown or compressed space with nearly no slowdown.  We leave as an open problem achieving linear or compressed space with $\Oh{1 / \tau}$ query time for variable $\tau$, or showing that this is impossible.

\begin{acks}
Many thanks to Patrick Nicholson for helpful comments.
\end{acks}

\bibliographystyle{authordate1}
\bibliography{wads13}

\end{document}